\documentclass[a4paper,11pt]{article}
\pdfoutput=1

\usepackage{jheppub}
\usepackage{natbib}
\usepackage{braket}
\usepackage{xcolor}
\usepackage{textcomp}
\usepackage{dcolumn}% Align table columns on decimal point
\usepackage{bm}% bold math
\usepackage{cleveref}
\usepackage{algorithm}
\usepackage{algpseudocode}
\usepackage[utf8]{inputenc}
\usepackage{tcolorbox}
\usepackage{todonotes}
\usepackage{amssymb,amsmath,amsbsy,amsthm,mathtools}
\usepackage{braket}
\usepackage{graphicx,color}
\usepackage{multirow}
\usepackage{hyperref}
\usepackage{subcaption}
\usepackage{adjustbox}
\usepackage{relsize} 
\usepackage{xcolor}
\usepackage{booktabs}
\usepackage{subcaption}
\usepackage{array}

\newcolumntype{P}[1]{>{\centering\arraybackslash}p{#1}}
\newcolumntype{M}[1]{>{\centering\arraybackslash}m{#1}}

\newtheorem{definition}{Definition}[section]
\newtheorem{proposition}{Proposition}[section]

\newtheorem{theorem}{Theorem}[section]
\newtheorem{corollary}{Corollary}[section]
\newtheorem{problem}{Main Problem}
\newtheorem{conjecture}{Conjecture}[section]

\newcommand{\mF}{\mathcal{F}}

\newcommand{\mM}{\mathcal{M}}

\newcommand{\mbE}{\mathbb{E}}

\newcommand{\mbZ}{\mathbb{Z}}

\DeclareMathOperator{\area}{area}

\usepackage{graphicx}

\title{Towards a complete classification of holographic entropy inequalities}

\author[a,b]{Ning Bao}
\author[a]{Keiichiro Furuya}
\author[a,c]{Joydeep Naskar}

\affiliation[a]{Department of Physics, Northeastern University, Boston, MA, 02115, USA}
\affiliation[b]{Computational Science Initiative, Brookhaven National Laboratory, Upton, NY 11973 USA}
\affiliation[c]{The NSF AI Institute for Artificial Intelligence and Fundamental Interactions, Cambridge, MA, U.S.A.}

\emailAdd{ningbao75@gmail.com}
\emailAdd{k.furuya@northeastern.edu}
\emailAdd{naskar.j@northeastern.edu}

\abstract{We propose a deterministic method to find all holographic entropy inequalities that have corresponding contraction maps and argue the completeness of our method. We use a triality between holographic entropy inequalities, contraction maps and partial cubes. More specifically, the validity of a holographic entropy inequality is implied by the existence of a contraction map, which we prove to be equivalent to finding an isometric embedding of a contracted graph. Thus, by virtue of the argued completeness of the contraction map proof method, the problem of finding all holographic entropy inequalities is equivalent to the problem of finding all contraction maps, which we translate to a problem of finding all image graph partial cubes. We give an algorithmic solution to this problem and characterize the complexity of our method. We also demonstrate interesting by-products, most notably, a procedure to generate candidate quantum entropy inequalities.}

\makeatletter
\gdef\@fpheader{}
\makeatother
\begin{document}

\maketitle

\section{Introduction}
Entropy inequalities constrain quantum states non-trivially. The first such inequality discovered was subadditivity, which constrains the entropies of the quantum states of a bipartite system and its components\cite{Araki:1970ba},
\begin{equation}\label{eq:sa}
    S(A)+S(B)\geq S(AB) \quad \quad \text{(Subadditivity)}.
\end{equation}
In addition to subadditivity, there are other inequalities that are obeyed by all quantum states. We list some of them below \cite{Araki:1970ba,ssa1973,weakmono2003,Headrick_2007},
%\begin{equation}
\label{eq:quantum-ineqs}
\begin{align}
    & S(A)+S(AB)\geq S(B) \quad \quad \text{(Araki-Lieb)} \label{eq:AL},\\
    & S(AB)+S(BC) \geq S(B)+S(ABC)  \quad \quad \text{(Strong subadditivity)}\label{eq:ssa},\\
    & S(AB)+S(BC) \geq S(A)+S(C) \quad \quad \text{(Weak monotonicity)}\label{eq:wm}.
\end{align}
%\end{equation}

Notably, the above inequalities are the only unconditional linear inequalities known to be obeyed by all quantum states. They enjoy extended permutation symmetries $S_{n+1}$. In particular, subadditivity \eqref{eq:sa} and Araki-Lieb inequality \eqref{eq:AL} are related by $S_3$ transformations together with symmetry by a purification, e.g., $S_{BO} = S_{A}$. Similarly, weak monotonicity \eqref{eq:wm} can be obtained from strong subadditivity \eqref{eq:ssa} by $S_4$.

A \emph{holographic entropy inequality}(HEI) is a non-trivial constraint on entanglement entropies of holographic states, a proper subset of all quantum states, at leading order. Holographic states are those quantum states in conformal field theory that are consistent with the existence of a semi-classical gravity dual, via the $AdS/CFT$ correspondence\cite{Maldacena1997}. More precisely, the Ryu-Takayanagi\cite{Ryu2006} formula, a formula that relates the entanglement entropy of a boundary subregion $A$ to a bulk geometric quantity, namely, the minimal surface $\gamma_A$ in the bulk homologous to $A$, i.e.,
\begin{equation}\label{eq:RTformula}
    S(A)=\frac{\area\gamma_A}{4G_N},
\end{equation}
applies for such states at leading order.

The simplest holographic entropy inequality that is not obeyed by all quantum states is the monogamy of mutual information(MMI)\cite{MMI2013} involving three parties, given by
\begin{equation}\label{eq:mmi}
    S(AB)+S(AC)+S(BC)\geq S(A)+S(B)+S(C)+S(ABC), \quad \quad \text{(MMI)}.
\end{equation}

MMI (\ref{eq:mmi}) taken together with subadditivity (\ref{eq:sa}) inequalities between subsystem pairs $AB$, $AC$ and $BC$ form the \emph{holographic entropy cone}(HEC)\cite{Bao:2015bfa} for three regions $\mathcal{C}_3$. The work of \cite{Bao:2015bfa} formalised the concept of the HEC and characterized it as a rational polyhedral cone, whose facets are identified as the minimal set of \emph{tight}\footnote{A HEI is a facet of the HEC iff there exists a codimension-1 set of linearly independent holographic entropy vectors for which the HEI saturates. We will use the terms facet HEI and tight HEI interchangeably.} inequalities for a given number of parties. In the same work, five new inequalities were discovered that describe the facets of the HEC for five regions $\mathcal{C}_5$. It was verified in \cite{HernandezCuenca:2019wgh} that these five inequalities (upto symmetry) taken together with the lower-party inequalities completely characterize $\mathcal{C}_5$. We give two examples of five-party inequalities below,
\begin{equation}\begin{split}
& S(ABC) + S(ABD) + S(ACE) + S(BCD) + S(BCE) \geq \\ & S(A) + S(BC) + S(BD) + S(CE) + S(ABCD) + S(ABCE),\label{eq:5.1}
\end{split} \end{equation}
\begin{equation}\begin{split}
& S(AD) + S(BC) + S(ABE) + S(ACE) + S(ADE) + S(BDE) + S(CDE) \geq \\ & S(A) + S(B) + S(C) + S(D) + S(AE) + S(DE) + S(BCE) + S(ABDE) + S(ACDE).\label{eq:5.2}
\end{split} \end{equation}

Over the past few years, several interesting directions emerged in the quest for discovering higher party HEIs to characterize $\mathcal{C}_n (n\geq 6)$. The work of \cite{Hernandez-Cuenca:2023iqh} discovered 1877 novel six-party inequalities, where the search is motivated by the observation that all known tight HEIs (except subadditivity\footnote{The subadditivity inequality is only balanced; it is neither superbalanced, nor expressible as a sum over (conditional) tripartite information.}) are superbalanced\cite{He:2020xuo} and can be written as a sum of tripartite and conditional tripartite information, conveniently expressible in the $I$-basis\cite{He:2019repackaged}. A parallel development was the construction of the \emph{holographic cone of average entropies}(HCAE)\cite{Czech:2021rxe} which motivated the discovery of two additional families of HEIs\footnote{The first family of cyclic inequalities was discovered in \cite{Bao:2015bfa} valid for any odd number of parties $p$. For example, inequality \ref{eq:5.1} is the $p=5$ case.} exploiting entanglement wedge nesting relations\cite{Czech:2023xed, Bao:2024vmy}. Such geometrization of HEIs provided excellent opportunities to study various qualitative features of (a family of) HEIs. Another method for studying HEIs involved the use of bit-threads \cite{Freedman_2016,headrick2022covariant, cui2019}. Several other works appeared in the recent years to study various aspects of the HEC\cite{Hernandez-Cuenca:2022marginal,Fadel:2021urx,Akers:2021lms, Czech:2024rco} and beyond the HEC, such as using hypergraphs\cite{Bao:2020zgx,Bao:2020uku}, topological link models\cite{Bao:2021gzu} and cycle flows\cite{He:2023rox}, studying gapped phases of matter\cite{Bao:2015boa,Naskar:2024mzi}, to name a few.\footnote{In all of these works, the entanglement entropy considered are \emph{bipartite entanglement measures} between a sub-region $X$ and its complement $\Bar{X}$ calculated for various multipartite arrangements of subregion $X$. Some other works concerning multipartite entropies calculated using \emph{multipartite entanglement measures} include \cite{Gadde:2022cqi,Gadde:2023zzj,Penington:2022dhr,Harper:2024ker}.}

Despite these developments, the task of fully characterizing all possible HEIs has proven elusive. Two key challenges stand in the way of this task. First, there is a combinatorial explosion associated with generating candidate inequalities. Second, one must develop methods, also often with similar combinatorial challenges, in order to prove a given conjectured inequality. One method of proving a HEI is to find a corresponding \emph{contraction map} (which we define in subsection \ref{subsec:contraction-map}). Until recently, the computational complexity of finding such maps was doubly exponential in the number of terms on the left hand side (LHS) using greedy approaches. The work of \cite{Bao:2024contraction_map}\footnote{Another complementary approach to speed-up the contraction map method can be found in \cite{Li:2022jji,AVIS202316}.} reduced this complexity significantly to a single exponential and demonstrated empirical evidence in favor of the completeness of the contraction map proof method.

In this paper, we take a step towards a \emph{complete}\footnote{We would like to think that the completeness argument is true, but a careful reader, should take our results to be valid for all HEIs having corresponding contraction maps.} classification of HEIs. However, we step aside from the conventional classification of HEIs into classes of HEC $\mathcal{C}_n$ of a fixed party number $n$. Instead, we classify HEIs into the classes $\mathcal{H}_M$ identified by the number of LHS terms (with unit coefficient) $M$. For example, the five party inequalities \ref{eq:5.1} and \ref{eq:5.2} both belong to $\mathcal{C}_5$ but they have different numbers of LHS terms. In our classification, they will correspond to two different classes, $\mathcal{H}_5$ and $\mathcal{H}_7$ respectively.

The organization of this paper is as follows. We outline the problem statement motivated by the holographic entropy cone program in \ref{subsec:phy-problem}. In section \ref{subsec:contraction-map}, we define the \textit{`proof by graph contraction'} method after reviewing the \textit{`proof by contraction'} method and introducing several key concepts, such as partial cubes and graph contraction maps, from graph theory. In section \ref{subsec:math-problem}, we rephrase the problem statement in binary expressions and graph theory. Moreover, we prove the equivalence of the problem statements by showing that the existence of a contraction map is a necessary and sufficient condition for that of a graph contraction map. In section \ref{subsec:solutions-graph}, we give a basic framework of the algorithm as a solution to the problem statement in graph theory and argue their completeness. More specifically, we show that all relevant graphs can be constructed from \emph{graph contraction} of a hypercube graph. %and a valid HEI corresponds to the resultant graph being a partial cube or isometrically hypercube embeddable. 
In section \ref{sec:read-HEI}, we give an algorithm to read off HEIs from a given contraction map (derived from a graph). Lastly, we discuss our results, their implications, and the future directions in section \ref{sec:discussions}.

\section{The problem statement: Find all holographic entropy inequalities}\label{sec:problem-statements}
In general, a $n$-party holographic entropy inequality (HEI) involving $n$ disjoint regions $[n]:=\{A_1, \cdots, A_n\}$ (and a purifier $O$), can be written in a basis of subregion entropies, constructed by their proper power set $P(\{A_1, \cdots, A_n\}) \backslash \emptyset$ containing $2^n-1$ possible elements. For example, in the case of the three regions $\{A,B,C\}$, the entropy basis is lexicographically written as
\begin{equation}
    \label{eq:basis}
    \{S_A, S_B, S_C, S_{AB}, S_{AC}, S_{BC}, S_{ABC}\}
\end{equation}
An inequality $\mathcal{Q}$ can be expressed as
\begin{equation}
    \label{eq:hei-q}
    \mathcal{Q}=\sum_{i=1}^{2^n-1} a_i S_{\tilde{X}_i} \geq 0,
\end{equation}
where $a_i$ is an integer coefficient associated with an entropy term $S_{\tilde{X}_i}$ corresponding to subregion $\tilde{X}_i$. Here, $a_i=0$ implies that the corresponding subregion entropy $S_{\tilde{X}_i}$ is absent from the inequality, whereas a positive (negative) $a_i$ implies that $S_{\tilde{X}_i}$ is present in the LHS (RHS). We can rearrange \eqref{eq:hei-q} to include only non-zero coefficients, to get an inequality of the form 
\begin{equation} \label{eq:genentineq}
    \sum_{i=1}^{l} c_i S_{X_i} \geq \sum_{j=1}^{r} d_j S_{Y_j},
\end{equation}
where $c_i$ and $d_j$ are positive integers. The number of non-zero coefficients on the LHS and RHS are $l$ and $r$, respectively. We would like to expand each coefficient, and re-write the inequality with repeating subregion entropies $S_{X_i}$ and $S_{Y_j}$ with unit coefficient, giving us
\begin{equation} \label{eq:genentineq-expand}
    \sum_{u=1}^{M} S_{L_u} \geq \sum_{v=1}^{N} S_{R_v}.
\end{equation}
where $L_u, R_v \in P(\{A_1, \cdots, A_n\}) \backslash \emptyset$ for $\forall u,v$. The number of terms (after expanding the coefficients) are $M$ and $N$ respectively, i.e.,

\begin{equation}
    \label{eq:lMrN}
    \sum_{i=1}^{l} c_i=M \quad \text{and} \quad \sum_{j=1}^r d_j= N.
\end{equation}

\subsection{Problem statement in physics}\label{subsec:phy-problem}
The thematic problem statement of our work is to find all holographic entropy inequalities.
\begin{tcolorbox}[colback=blue!5!white, colframe=blue!75!black, title=Thematic Problem, fonttitle=\bfseries]
Find all holographic entropy inequalities that have corresponding contraction maps.
\end{tcolorbox}

We begin by fixing the number of LHS terms of our inequality (not the terms), i.e, we are fixing $M$. 
Since $M$ is arbitrary, finding all possible HEIs for a given $M$ allows us to find all possible HEIs. Now we are ready to define our main problem statement.
\begin{problem}[Physics statement]\label{prob:phys}\
Consider a set of subregions. Given a convex composition of entropies with positive integer coefficients $c_i\geq 0 $ for $\forall i=1,\cdots,l$, 
     \begin{equation}\label{eq:lhs-mp}
         c_1 S_{X_1} +\cdots + c_{l} S_{X_{l}}.
     \end{equation}
     Find all the possible convex combinations of entropies with positive integer coefficients\footnote{It suffices to consider integer coefficients because the HEC is a rational polyhedral cone\cite{Bao:2015bfa}.} $d_j \geq 0 $ for $\forall j=1,\cdots,r$,
     \begin{equation}
         d_1 S_{Y_1} +\cdots + d_r S_{Y_{r}},
     \end{equation}
     such that
     \begin{equation}\label{eq:ineq-mp}
         c_1 S_{X_1} +\cdots + c_l S_{X_{l}} \geq d_1 S_{Y_1} +\cdots + d_r S_{Y_{r}}
     \end{equation}
    is a valid HEI.%\footnote{Note that due to results from \cite{Bao:2015bfa} showing that the cone is rational and polyhedral, only positive integer coefficients must be considered.}
\end{problem}
We can simplify this problem a little bit to fit our methods by expanding the coefficients of the $l$ terms on the LHS, such that there are $M$ entropy terms on the LHS with unit coefficient (see equation \ref{eq:lMrN}).\footnote{Note that when the coefficients $c_i>1$, there are repetition of terms after expanding.} We have repeating terms in this case, however, it suffices to stick to unit coefficients with non-repeating LHS terms by introducing more parties\footnote{All such LHS with $c_i>1$ can be uplifted from an $n$-party expression to a $(n+k)$-party expression, having non-repeating terms with unit coefficients. One can always do the reduction by trivializing those $k$ parties, and go back to the $n$-party expression.}. From now on, all the LHS we will refer to have unit coefficients. We thus reformulate our main problem \ref{prob:phys} in the box \ref{prob:phys-simple}. Note that this reformulation of the problem is equivalent to the original one.\\
\setcounter{problem}{0}
\begin{tcolorbox}[colback=blue!5!white, colframe=blue!75!black, title=Main Problem 1 (Unimodular), fonttitle=\bfseries]
\begin{problem}[Physics statement: Unimodular] \label{prob:phys-simple}
Given $M$ LHS terms with unit coefficients, 
\begin{equation}\label{eq:LHS-unit-coeff}
    S_{L_1} +\cdots + S_{L_{M}}.
\end{equation}
Find all possible RHS terms, such that the following inequality holds
\begin{equation}
    S_{L_1} +\cdots + S_{L_M} \geq S_{R_1} +\cdots + S_{R_N},
\end{equation}
for some $N$ with unit coefficients.
    
\end{problem}
\end{tcolorbox}

We will give an algorithmic recipe to answer this problem in the language of graph theory. The LHS comprising $M$ terms (\ref{eq:LHS-unit-coeff}) can be mapped to a hypercube graph $H_M$ of $2^M$ binary bitstrings.
One then performs graph contractions\cite{Erciyes2018} on this hypercube graph $H_M$, resulting in a contracted graph $G$. Since the number of possible graph contractions is finite (albeit, large), the number of contracted graphs is also finite. Isometric hypercube embeddability of this resultant graph $G$ into a hypercube $H_N$, should it exist, would define a contraction map from $\{0,1\}^M$ to $\{0,1\}^N$. Given a contraction map, one can find the HEIs corresponding to the contraction map by assigning boundary conditions. Thus, the problem of finding all such HEIs can be framed as a problem of finding all contraction maps. We show the problem of finding all contraction maps is equivalent to determining which subset of graphs $G$ are partial cubes, where the graphs $G$ are obtained by graph contractions starting from $H_M$.

In the rest of this paper, we will prove that this method is sufficient to generate all possible HEIs corresponding to contraction maps. We will give a more formal mathematical presentation of the problem statements below after introducing relevant definitions.  We summarize our algorithmic method in algorithm \ref{alg:all-hei} in subsection \ref{subsec:solutions-graph}.

\subsection{Contraction maps and graph contraction maps}
\label{subsec:contraction-map}
The \textit{`proof by contraction'} method was proposed in \cite{Bao:2015bfa} to prove the validity of a candidate HEI (\ref{eq:genentineq}). This method relies on the equivalence between holographic geometries and graphs, labelled by bitstrings. In principle, this method is based on the inclusion/exclusion of bulk regions and extracts the information of contributions of the RT surfaces to the convex compositions of the RT entropies and encodes them into bitstrings. 

A particularly relevant set of bitstrings are those labeling regions adjacent to boundary subsystems; for each $i\in[n+1]$,
\begin{equation}\label{eq:occurence-bitstrings}
    x^u_{A_i} = \begin{cases}
        1, \qquad \text{if} \quad A_i\in L_u \\
        0, \qquad \text{otherwise,}
    \end{cases}
\end{equation}
where the ($n+1$)-th bitstring is all zeroes and is associated with the purifier $O$. These bitstrings are the boundary conditions (also known as \emph{occurrence bitstrings}) and can be analogously defined for RHS subsystems using bitstrings $y\in\{0,1\}^N$.

A contraction map is a map between bitstrings from $\{0,1\}^M$ to $\{0,1\}^N$ such that the Hamming distance between every pair of bitstrings in $\{0,1\}^M$ is greater than or equal to the Hamming distance between their images in $\{0,1\}^N$. A contraction map learns about a particular inequality through the boundary conditions. We will now define these ideas more formally below.

\begin{definition}[Hamming distance]
The Hamming distance between two bitstrings of length $L$ is defined as the number of positions where the bitstrings differ, i.e.,
\begin{equation}
    d_H(x,x') = \sum_{u=1}^L |x^u - x'^{u}|.
\end{equation}
where $x^u$ is the $u$-th bit.
\end{definition}

\begin{theorem}[`Proof by contraction']\label{thm:proofbycontraction}\cite{Bao:2015bfa}
    Let $f:\{0,1\}^M \to \{0,1\}^N $ be a contraction map, i.e.,
    \begin{equation} \label{eq:contractioncondition}
        d_H(x ,x') \geq d_H(f(x) ,f(x')) ,\; \forall x,x'\in \{0,1\}^M.
    \end{equation}
    If  $f(x_{A_i})  = y_{A_i}$ for $\forall i\in \{1,\cdots, n+1\}$, then (\ref{eq:genentineq-expand}) is a valid $n$-party HEI.
\end{theorem}

We can rephrase the language of bitstrings into its natural habitat in graph theory by reinterpreting the bitstrings $\{0,1\}^M$ to be a $M$-dimensional hypercube, which is mapped to a subset of $\{0,1\}^N$ satisfying the contraction condition (\ref{eq:contractioncondition}) to be a subgraph of $N$-dimensional hypercube.

Consider two unit-weighted and undirected graphs, $G_1=(V_1,E_1)$ and $G_2=(V_2,E_2)$. In general, the graph map $\Phi:=(\phi^V,\phi^E)$ from $G_1$ to $G_2$ consists of a map $\phi^V:V_1 \to V_2$ and $\phi^E:E_1\to E_2$. A weak graph homomorphism $\Phi$ is a graph map that either preserves the adjacency $(\phi^V(v),\phi^V(v'))\in E_2$ or contracts the vertices $\phi^V(v)=\phi^V(v')$ for $v,v'\in V_1$ and $(v,v')\in E_1$\cite{hammack2011}. Hence, it does not increase the graph distance. A graph distance $d(v,v')$ is the shortest path or the minimum number of edges connecting between the vertices $v,v'\in V$\cite{Bondy-Murty}.

We define a \textit{graph contraction map} to be a weak graph homomorphism to construct two alternative formulations below. Going forward, we will only use the term \textit{graph contraction map} for clarity.

\begin{definition}[A graph contraction map/weak graph homomorphism]\

    For a graph $G=(V,E)$, a graph map $\Phi: G\to \Phi(G)$ is a graph contraction map if
    \begin{equation}
        d(v,v') \geq d(\phi^V(v),\phi^V(v')), \;\forall v,v'\in V.
    \end{equation}
    
\end{definition}

From theorem \ref{thm:proofbycontraction}, the domain of graph contraction maps is always an $M$-dimensional hypercube $H_M$. Moreover, by virtue of the bitstrings picture, the images $\Phi(H_M)$ of the graph contraction maps $\Phi$ are subgraphs of the $N$-dimensional hypercube $H_N$. In fact, the subgraphs $\Phi(H_M)$ should be \textit{isometrically} embeddable in $H_N$\footnote{A graph $G$ is isometrically embeddable to another graph $G'$ if the adjacency and the graph distance of all pairs of vertices of $G$ are preserved.}. Such graphs are called \textit{partial cubes}. We define a partial cube and its isometric dimension below. 
\begin{definition}[Partial cube\cite{Ovchinnikov,WINKLER1984221}]
    A graph $G$ is a partial cube if $G $ is isometrically embeddable to a $D$-dimensional hypercube graph $H_D$.  
\end{definition}

\begin{definition}[Isometric dimension \cite{Ovchinnikov}]
    The isometric dimension $idim(G)$ of a partial cube $G$ is the minimum dimension of a cube in which $G$ is isometrically embeddable.
\end{definition}
A $M$-dimensional hypercube graph is the trivial partial cube of isometric dimension $idim(H_M)=M$. We denote $\Phi: H_M\to H_N$ if a graph map is from a $M$-dimensional hypercube $H_M$ to a partial cube with $idim(\Phi(H_M))=N$.

Correspondingly to the occurrence bitstrings(\ref{eq:occurence-bitstrings}), we define \textit{occurence vertices} or \textit{boundary vertices} as 
\begin{equation}\label{eq:occurence-vertices}
    \big(\iota_M(v_{A_i})\big)^u = \begin{cases}
        1, \qquad \text{if} \quad A_i\in L_u \\
        0, \qquad \text{otherwise,}
    \end{cases}
\end{equation}
where $\iota_M:V_M \to \{0,1\}^M$ is an isometry\footnote{All the \textit{isomteries} in this paper are \textit{path isometries}, which are distance-preserving maps between metric spaces. In particular, they map between a metric space of bitstrings $\{0,1\}^{M,N}$ equipped with a Hamming distance $d_H$ and a metric space of vertices $V_{M,N}$ equipped with a graph distance $d$. }, i.e.,
\begin{equation}
    d(v,v') = d_H(\iota_M(v),\iota_M(v')) 
\end{equation}
for $v,v'\in V_M$.

Now, we are ready to formulate the \textit{`proof by graph contraction'}, as follows.
\begin{theorem}[`Proof by graph contraction']\label{thm:proofbygraphcontraction}\cite{Bao:2015bfa}
    Let $\Phi:H_M=(V_M,E_M) \to H_N =(V_N,E_N)$ be a graph contraction map, i.e.,
    \begin{equation}
        d(v ,v') \geq d(\phi^V(v) ,\phi^V(v')) ,\; \forall v,v'\in V_M.
    \end{equation}
    If $\phi^V(v_{A_i})  = v^N_{A_i}$ for $\forall i\in \{1,\cdots, n+1\}$ where $v_{A_i}\in V_M$ and $v^N_{A_i} \in V_N$, then (\ref{eq:genentineq}) is a valid HEI.
\end{theorem}

\subsection{Problem statements in mathematics: bitstrings and graph theory}
\label{subsec:math-problem}
In this subsection, we define and discuss two alternative formulations of the main problem \ref{prob:phys-simple} as problems \ref{prob:binary} and \ref{prob:graph} and prove their equivalence. In the language of bitstrings, the main problem \ref{prob:phys-simple} is rephrased as follows.

\begin{tcolorbox}[colback=blue!5!white, colframe=blue!75!black, title=Main Problem 2 (Bitstrings), fonttitle=\bfseries]
\begin{problem}[Mathematical statement: Bitstrings]\label{prob:binary}
    
    Given a complete set of bitstrings\footnote{A complete set of bitstrings is the set of all possible $2^M$ binary bitstrings of length $M$.} $\{0,1\}^M$ for a fixed $M\in \mbZ_+$, find the set of all the possible contraction maps $\mF = \{f| f:\{0,1\}^M\to\{0,1\}^N\}$ for some $N\in \mbZ_+$, where $f$ maps the complete set of bitstrings in $\{0,1\}^M$ to a subset of bitstrings in $\{0,1\}^N$.
    
\end{problem}
\end{tcolorbox}

The equivalence between the problem \ref{prob:phys-simple} and the problem \ref{prob:binary} can be conjectured by theorem \ref{thm:proofbycontraction} and the \textit{completeness of `proof by contraction'}\cite{Bao:2024contraction_map} which argues that there exists at least one contraction map for every valid HEI.

In the graph-theoretic reformulation, we have the following problem statement.

\begin{tcolorbox}[colback=blue!5!white, colframe=blue!75!black, title=Main Problem 3 (Graph theory), fonttitle=\bfseries]
\begin{problem}[Mathematical statement: Graph theory]\label{prob:graph}
%$\mF_G = \{\Phi| \Phi : H_l \to \Phi(H_l),  \Phi(H_l) \hookrightarrow H_r \}$

    Given a $M$-dimensional hypercube $H_M$ for a fixed $M\in \mbZ_+$, find the set of all the possible graph contraction maps $\mF_G = \{\Phi|\Phi:H_M\to H_N\}$ such that the images $Im\mF_G = \{\Phi(H_M)\}$ are partial cubes of $dim(\Phi(H_M))=N$ for some $N\in \mbZ_+$.
    
\end{problem}
\end{tcolorbox}

We prove that problem \ref{prob:binary} is equivalent to the problem \ref{prob:graph} by the following proposition. %Therefore, all the three main problems \ref{prob:phys-simple}, \ref{prob:binary}, and \ref{prob:graph} are equivalent.

\begin{proposition}[Equivalence between contraction maps and graph contraction maps]\label{pro:equivalence-bit-graph}\ 
    Given $M\in \mbZ_+$, there exists a contraction map $f:\{0,1\}^M\to \{0,1\}^N$ for some $N\in \mbZ_+$ if and only if there exists a graph contraction map $\Phi:H_M \to H_N$.
\end{proposition}
\begin{proof}

    $(f\rightarrow \Phi)$ We first show that if $f:\{0,1\}^M\to \{0,1\}^N$ is a contraction map given $M$ for some $N$, there exists a graph contraction map $\Phi:H_M=(V_M,E_M)\to H_N=(V_N,E_N)$.

    Let us define a path isometry $\tilde{\iota}_J: \{0,1\}^J \to V_J$ for $J\in \mbZ_+$ such that
    \begin{equation}
        d_H(x,x') = d(\tilde{\iota}_J(x),\tilde{\iota}_J(x')), \forall x,x'\in\{0,1\}^J.
    \end{equation}
    Then, we construct a graph map $\Phi_f:=(\phi^V_f,\phi^E_f)$ from a contraction map $f$ using the isometries $\tilde{\iota}_M:\{0,1\}^M\to V_M$ and $\tilde{\iota}_N:\{0,1\}^N\to V_N$. That is, we define $\phi^V_f$ and $\phi^E_f$ such that, for $x\in\{0,1\}^M$ and $f(x) \in \{0,1\}^N$,
    \begin{itemize}
        \item $\phi^V_f\circ \tilde{\iota}_M(x)= \tilde{\iota}_N \circ f(x)$ 
        \item $\phi^E_f(\tilde{\iota}_M(x),\tilde{\iota}_M(x')) =(\tilde{\iota}_N\circ f(x), \tilde{\iota}_N \circ f(x'))$ if
        \begin{equation}
            d_H(f(x),f(x'))=1,
        \end{equation}
        however $\phi^E(E)/\phi^E_f(\tilde{\iota}_M(x),\tilde{\iota}_M(x'))$\footnote{We denote $G/e$ as deleting an edge $e$ from a graph $G$. %See appendix \ref{app:graph-operations} for more notations on graph operations. 
        We have this condition to remove self-loops from the image graph $\Phi(H_M)$.} if 
        \begin{equation}
            d_H(f(x),f(x'))=0.
        \end{equation}
    \end{itemize}
    By construction, $\phi^V_f \circ \tilde{\iota}_M(x) \in V_N$ and $\phi^E_f(\tilde{\iota}_M(x),\tilde{\iota}_M(x')) \in E_N$. Thus, $\Phi_f$ is a graph map from $H_M$ to $H_N$. In other words, the image $\Phi_f(H_M)$ is always a partial cube of $dim(\Phi_f(H_M))=N$.
    \begin{figure}
       \centering
       \includegraphics[width=0.3\linewidth]{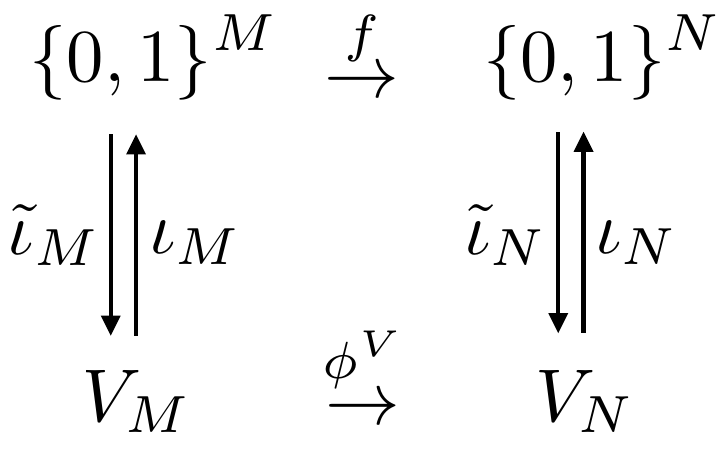}
       \caption{\small{The diagram of the maps in the proof of proposition \ref{pro:equivalence-bit-graph}. The isometries $\iota_{M,N}$ and $\tilde{\iota}_{M,N}$ map between bitstrings $\{0,1\}^{M,N}$ and the set of vertices $V_{M,N}$ of hypercubes $H_{M,N}$. $f$ is a contraction map. $\phi^V$ is a vertex map from $H_M$ to $H_N$. }}
       \label{fig:enter-label}
    \end{figure}
    $\Phi_f$ is a graph contraction map because, for all $ x,x'\in \{0,1\}^M$,
    \begin{equation}
    \begin{split}
        d(\tilde{\iota}_M(x),\tilde{\iota}_M(x')) &= d_H(x,x')\\
        &\geq d_H(f(x),f(x')) \\
        &= d(\phi^V_f\circ \tilde{\iota}_M(x),\phi^V_f\circ\tilde{\iota}_M(x')). \\
    \end{split}
    \end{equation}
    We applied the definition of isometry $\tilde{\iota}_M$ for the first equality. The second inequality is due to the contraction property of $f$. The isometry $\tilde{\iota}_N$ and the definition of $\phi^V_f$ implies the last equality. Therefore, $\Phi_f:H_M\to H_N$ is a graph contraction from $H_M$ to $H_N$.

    $(f \leftarrow \Phi)$ We next show that if a graph map $\Phi$ is a graph contraction map $\Phi:H_M=(V_M,E_M)\to H_N=(\phi^V(V_M),\phi^E(E_M))$ given $M$ for some $N$, there exists a contraction map $f:\{0,1\}^M\to \{0,1\}^N$. 

    Similarly, we define a path isometry $\iota_J:V_J\to \{0,1\}^J$ for $J\in \mbZ_+$ such that
    \begin{equation}
        d_H(\iota_J(v),\iota_J(v')) = d(v,v'), \forall v,v'\in V_J
    \end{equation}
    For $\iota_M:V_M\to \{0,1\}^M$ and $\iota_N:V_N\to \{0,1\}^N$, a contraction map $f_\Phi$ is constructed such that
    \begin{equation}
        f_\Phi \circ \iota_M (v)= \iota_N \circ \phi^V (v).
    \end{equation}
    Then, we see that $f_\Phi:\{0,1\}^M \to \{0,1\}^N$ simply because $\iota_M(v) \in \{0,1\}^M$ and $\iota_N\circ \phi^V (v) \in \{0,1\}^N$. %Note that $\phi^V (v) \in V_N$ by the definition of the graph contraction map $\Phi:H_M \to H_N$.

    For all $ v,v' \in V_M$, 
    \begin{equation}
    \begin{split}
         d_H(\iota_M (v),\iota_M (v')) &=d(v,v') \\
         &\geq d(\phi^V(v),\phi^V(v')) \\
         &= d_H(f_\Phi\circ \iota_M (v), f_\Phi\circ \iota_M (v')).\\
    \end{split}
    \end{equation}
    Therefore, $f_\Phi:\{0,1\}^M \to \{0,1\}^N$ is a contraction map.

\end{proof}

\subsection{An algorithmic solution to the main problems}\label{subsec:solutions-graph}

In this subsection, we provide an algorithmic solution to the main problem \ref{prob:graph}. The key ingredients of our solution are i) graph contraction maps and ii) partial cubes or isometric hypercube embeddability. We use the following three algorithms to realize a graph contraction map from a $M$-dimensional hypercube graph $H_M=(V_M,E_M)$ to a partial cube $G$ of $dim(G)=N$, as follows: 1. partition generator\cite{er1988}, 2. graph contraction\cite{Erciyes2018}, and 3. partial cube identifier\cite{Eppstein_2011}, see algorithm \ref{alg:all-hei} and figure \ref{fig:graph-to-contractionmap}.

The \texttt{Partition generator} algorithm in \cite{er1988} recursively generates all partitions of a given set, a set $V_M$ of vertices of hypercube graph $H_M$ in our case\footnote{Note that taking all partitions of vertices can reproduce all graph contractions. However, there is a scope of algorithmic improvement at this step that considers the relevant and non-redundant partitions only.}. The number of partitions is given by \textit{Bell number} $B_\alpha$ of a set of cardinality $\alpha$, i.e.,
\begin{equation}
    B_\alpha=\sum_{\beta=0}^\alpha P(\alpha,\beta)
\end{equation}
where $P(\alpha,\beta)$ is the \textit{Stirling number of the second kind} computing the number of partitions of the set of cardinality $\alpha$ into $\beta$ non-empty subsets. 

The \texttt{Graph contraction} algorithm \cite{Erciyes2018} has been studied to provide parallel algorithms to solve graph-associated problems, for instance, \cite{Erciyes2018,LOMBARDI2022429,Guattery1992}. We apply the algorithm to a hypercube graph $H_M$ to generate contracted graphs based on the partitions generated above.

However, it turns out that the output graphs generated by the \texttt{Graph contraction} algorithm are not necessarily isometrically hypercube embeddable or partial cubes. So, we apply a polynomial time algorithm \cite{Eppstein_2011} to check whether the graph is a partial cube. We call this algorithm \texttt{Partial cube identifier}.

We now describe algorithm \ref{alg:all-hei} in detail. For a hypercube graph $H_M=(V_M,E_M)$ with a path isometry $\iota_M:V_M\to \{0,1\}^M$, let us call $(H_M,\iota_M)$ as an \textit{initial graph data}.
\begin{enumerate}
    \item  \texttt{Partition generator}\cite{er1988}: Given an initial graph data $(H_M,\iota_M)$, generate $B_{2^M-1}$ partitions, $ p^\sigma_w$ for $\sigma=1,\cdots, B_{2^M-1}$, of $V_M$, i.e., 
    \begin{equation}
         \bigcup_{w=0}^{\beta_\sigma} V_{p^\sigma_w}=V_M   
    \end{equation}
    where $V_{p^\sigma_w}$ are disjoint subsets of vertices for the $\sigma$-th partition $p^\sigma_w$, and $w=0,\cdots, \beta_\sigma$ labels the $\beta_\sigma$ non-empty subsets in the $\sigma$-th partition. We fix $V_{p^\sigma_0} = \{0\}^M= \{0\cdots0\}$ as a single element subset\footnote{In our gauge choice, we choose $\{0\}^M$. One can choose, for instance, $\{1\}^M$ instead.}\footnote{For clarification of notation, $\{0\}^M$ is the binary bitstring of length $M$ of all $0$s. E.g., $\{0\}^3:=\{000\}$.}.

    Example: Given $H_{M=3}=(V_{M=3},E_{M=3})$ and $\iota_{M=3}:V_{M}\to \{0,1\}^M$ such that 
    \begin{equation}\label{eq:example-partitions}
    \begin{split}
        &\iota_{M}(v_0)=000,\;\iota_{M}(v_1)=001,\;\iota_{M}(v_2)=010, \;\iota_{M}(v_3)=100,\\
        &\iota_{M}(v_4)=011,\;\iota_{M}(v_5)=101,\;\iota_{M}(v_6)=110,\;\iota_{M}(v_7)=111.
    \end{split}
    \end{equation}
    Then, we have a list of partitions, some of which are shown below,
    \begin{equation}
    \begin{split}
        & \sigma = 1: V_{p_0^1} =\{v_0\},\;V_{p_1^1} =\{v_1,v_2,v_3,v_4,v_5,v_6,v_7\}\\
        & \sigma=2: V_{p_0^2} =\{v_0\},\;V_{p_1^2} =\{v_1,v_2,v_3,v_4,v_5,v_6\}, \;V_{p_2^2} =\{v_7\} \\
        & \;\;\; \vdots \\
        & \sigma\text{-th} \;\;\; : V_{p_0^\sigma}=\{v_0\}, V_{p_1^\sigma}=\{v_1,v_2,v_3,v_7\}, V_{p_2^\sigma}=\{v_4\},V_{p_3^\sigma}=\{v_5\},V_{p_4^\sigma}=\{v_6\}\\
        & \;\;\; \vdots \\
    \end{split}
    \end{equation}

    \item \texttt{Graph contraction} \cite{Erciyes2018}: Construct a graph $G_\sigma=(V_\sigma,E_\sigma)$ of, for instance, the $\sigma$-th partition following the three steps below.
    \begin{enumerate}
        \item[i)] Choose a partition of $V_M$, e.g., the $\sigma$-th partition, or $\bigcup_{w=0}^{\beta_\sigma} V_{p^\sigma_w} = V_M $.
        \item[ii)] Identify all the vertices in the subset $V_{p^\sigma_w}$ as a vertex labeled with $\chi_{p^\sigma_w}$ for $w = 1,\cdots, \beta_\sigma$, i.e., for every $w = 1,\cdots, \beta_\sigma$,
        \begin{equation}
            (V_M\backslash V_{p^\sigma_w})\cup \{\chi_{p^\sigma_w}\}\text{ and } E_M\backslash E_{p_w^\sigma}
        \end{equation}
        where $E_{p_w^\sigma}: =\{(v,v')\in E_M| v,v' \in V_{p^\sigma_w}\}$.
        
        Thus, 
        \begin{equation}
            V_\sigma =\{\chi_{p^\sigma_w}|(V_M\backslash V_{p^\sigma_w})\cup \{\chi_{p^\sigma_w}\}, \forall w = 1,\cdots, \beta_\sigma\} 
        \end{equation}
        \begin{equation}\label{eq:graph-contraction-multiedge}
        \begin{split}
            \mbE_\sigma = \{\;(\!(\chi_{p^\sigma_w},\chi_{p^\sigma_{w'}})\!)\;| & \chi_{p^\sigma_w},  \chi_{p^\sigma_{w'}}\in V_\sigma\text{ s.t. } \\
            &\;\; (v,v')\in E_M,  v\in V_{p^\sigma_w} ,v'\in V_{p^\sigma_{w'}}; \\
            &\forall w = 1,\cdots, \beta_\sigma\}.\\
        \end{split}
        \end{equation}
        \eqref{eq:graph-contraction-multiedge} is a set of edges connecting the vertices $\chi_{p^\sigma_w},\chi_{p^\sigma_{w'}}\in V_\sigma$ if the vertices $v\in V_{p^\sigma_w}$ and  $v'\in V_{p^\sigma_{w'}}$ were adjacent to each other in $H_M$, i.e., $(v,v')\in E_M$. The double parentheses, such as $(\!(\chi_{p^\sigma_w},\chi_{p^\sigma_{w'}})\!)$, denote that there could be more than a single edge connecting the vertices $\chi_{p^\sigma_w}$ and $\chi_{p^\sigma_{w'}}$ as opposed to a single parenthesis, for instance, $(\chi_{p^\sigma_w},\chi_{p^\sigma_{w'}})$, representing a single edge between the vertices.
        
        \item[iii)]
            If there are more than a single edge $(\!(\chi_{p^\sigma_w},\chi_{p^\sigma_{w'}})\!)$ between the vertices in $V_\sigma$, remove the extra edges so that there is only a single edge $(\chi_{p^\sigma_w},\chi_{p^\sigma_{w'}})$, i.e., 
            \begin{equation}
            \begin{split}
                \mbE_\sigma \to E_\sigma = \{\;(\chi_{p^\sigma_w},\chi_{p^\sigma_{w'}})\;| & \chi_{p^\sigma_w},  \chi_{p^\sigma_{w'}}\in V_\sigma, \forall w = 1,\cdots, \beta_\sigma \}.\\
            \end{split}
            \end{equation}

        %\item Remove edges $(v,v')$ between $v,v'\in V_{p^\sigma_w}$.
        %\item Merge multiple parallel edges between $V_{p^\sigma_w}$ and $V_{p^\sigma_{w'}}$.
    \end{enumerate}

    \item \texttt{Partial cube identifier}\cite{Eppstein_2011}: Check if $G_\sigma$ is a partial cube. 
    \begin{itemize}
        \item If so, the algorithm determines the isometric dimension $idim(G_\sigma)=N$ for some $N\in \mbZ_+$ and an isometry $\iota_N : V_\sigma \to \{0,1\}^N$. Thus, we obtain a \textit{final graph data} $(G_\sigma, \iota_N)$ depending on a choice of partition $\sigma$.
        \item Otherwise, discard the non-partial cubes $G_\sigma$. 
    \end{itemize}
\end{enumerate}

\begin{figure}[t]
    \centering
    \includegraphics[width=0.9\linewidth]{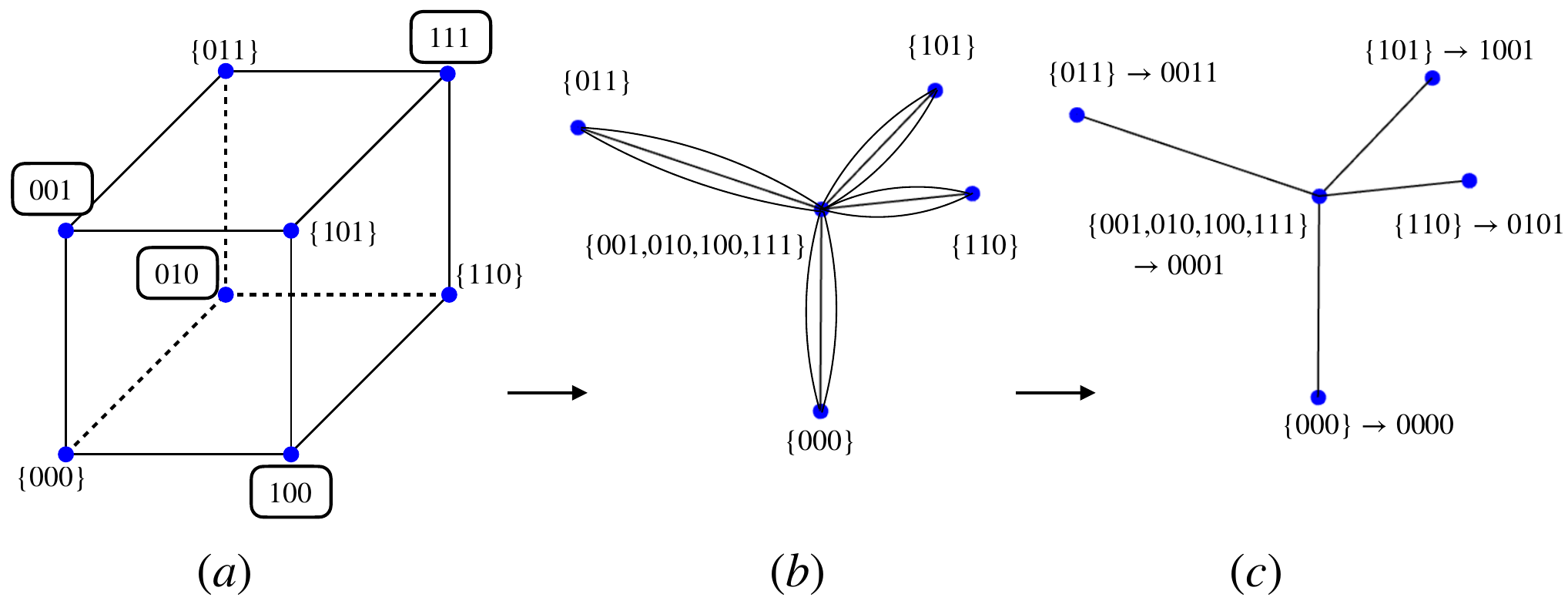}
    \caption{\small{Reading off a contraction map from the graph contraction mapping from $H_{M=3}$ to $H_{N=4}$. The blue dots are vertices. The black lines are the edges connecting the vertices. (a) Each vertex is labeled with $\{0,1\}^3$. The $\sigma$-th partition, for example, $\iota_M(V_{p_0^\sigma})=\{\iota_M(v_0)=000\}, \iota_M(V_{p_1^\sigma})=\{\iota_M(v_1)=001,\iota_M(v_2)=010,\iota_M(v_3)=100,\iota_M(v_7)=111\} ,\iota_M(V_{p_2^\sigma})=\{\iota_M(v_4)=011\},\iota_M(V_{p_3^\sigma})=\{\iota_M(v_5)=101\},\iota_M(V_{p_4^\sigma})=\{\iota_M(v_6)=110\}$ is chosen. The vertices labeled with the bitstrings in $\{001,010,100,111\}$ are enclosed by a rounded square. (b) After identifying the vertices based on the choice of the partition, there are three edges between every pair of the adjacent vertices $\chi_{p^\sigma_w}\in V_\sigma$ in the new graph $(V_\sigma, \mbE_\sigma)$. We obtain the graph $G_\sigma=(V_\sigma, E_\sigma)$ in (c) by removing two edges between every pair of the adjacent vertices. (c) The graph contraction with the choice of partition generates a star graph, a partial cube of isometric dimension $N=4$. Every vertex $\chi_{p^\sigma_w}\in V_\sigma$ gets labeled with a bitstring in a subset of $\{0,1\}^4$, e.g., $\Big\{\iota_{N=4}(\chi_{p^\sigma_0}) = 0000,\iota_4(\chi_{p^\sigma_1}) =0001, \iota_4(\chi_{p^\sigma_2}) =0011,\iota_4(\chi_{p^\sigma_3}) =1001,\iota_4(\chi_{p^\sigma_4}) =0101\Big\}$.}}
    \label{fig:graph-to-contractionmap}
\end{figure}

Once we generate all partial cubes with all possible isometries\footnote{For a partial cube $G=(V,E)$ of isometric dimension $idim(G)=N$ generated by algorithm 1, there are at most $|V|\times N!$ distinct isometries $\iota_N$.}, we can read off a contraction map from the initial graph data $(H_M,\iota_M)$ and the final graph data $(G_\sigma,\iota_N)$. That is, for a choice of partition, a contraction map $f_{(\sigma,\iota_N)}:\{0,1\}^M\to\{0,1\}^N$ is determined\footnote{See, for example, figure \ref{fig:graph-to-contractionmap}.} such that,  for $\forall v \in V_{p^\sigma_w}$ and $\forall w=1,\cdots, \beta_\sigma$, 
\begin{equation}\label{eq:graph-contractionmap}
    f_{(\sigma,\iota_N)}\circ \iota_M(v) = \iota_N(\chi_{p^\sigma_w})
\end{equation}
where $\chi_{p^\sigma_w} \in V_\sigma$.

Recall the example we discussed in eq. \ref{eq:example-partitions} (and in figure \ref{fig:graph-to-contractionmap}). We have,
\begin{equation}
\begin{split}
    f_{(\sigma,\iota_N)}\circ \iota_M(v_0) =  \iota_N(\chi_{p^\sigma_0}) = 0000,&\;f_{(\sigma,\iota_N)}\circ \iota_M(v_4) =  \iota_N(\chi_{p^\sigma_2}) = 0011 \\ 
    f_{(\sigma,\iota_N)}\circ \iota_M(v_5) =  \iota_N(\chi_{p^\sigma_3}) = 1001,&\;f_{(\sigma,\iota_N)}\circ \iota_M(v_6) =  \iota_N(\chi_{p^\sigma_4}) = 0101 \\ 
    f_{(\sigma,\iota_N)}\circ \iota_M(v_1)=f_{(\sigma,\iota_N)}\circ \iota_M(v_2)=f_{(\sigma,\iota_N)}&\circ \iota_M(v_3) = f_{(\sigma,\iota_N)}\circ \iota_M(v_7) =  \iota_N(\chi_{p^\sigma_1}) = 0001.
\end{split}
\end{equation}
See table \ref{tab:mmi-map} for the corresponding contraction map.

\begin{algorithm}
    \caption{A road map to generate all HEIs for a fixed number of LHS terms.}
    \label{alg:all-hei}
    \begin{algorithmic}[1]
        \Procedure{All HEIs for fixed $M$}{}
            \State Construct an initial graph data $(H_M,\iota_M)$.% hypercube $H_M=(V_M,E_M)$.
            \State Generate all partitions $\{\sigma\}$ of hypercube vertices using \texttt{Partition generator}.
            \For{  $\sigma \in \{\sigma\}$ }
            \State Find the contracted graph $G_\sigma=(V_\sigma,E_\sigma)$ using \texttt{Graph contraction}.
            \State Check whether $G_\sigma$ is a partial cube using \texttt{Partial cube identifier}.
            \If{$G_\sigma$ is a partial cube}
            \State Obtain a final graph data $(G_\sigma, \iota_N)$.
            \State Generate HEIs from $(G_\sigma,\iota_N)$ using \texttt{Contraction Map to Inequalities}.
            \EndIf
            \EndFor
        \EndProcedure
    \end{algorithmic}
\end{algorithm}

The procedure generates all possible contraction maps because, first, it generates all possible partial cubes from $H_M$. Moreover, the contraction maps depend on both a choice of partition and isometry. The graph data $(G_\sigma, \iota_N)$ and $(G_{\sigma}, \iota'_N)$ that differ only by the isometries $\iota_N \neq \iota_N'$ can generate possibly two distinct contraction maps. In addition, even if two partial cubes $G_\sigma$ and $G_{\sigma'}$ of the graph data $(G_\sigma, \iota_N)$ and $(G_{\sigma'}, \iota_N)$ for $\sigma\neq \sigma'$ are graph isomorphic and the isometries $\iota_N$ are identical, they could also construct different contraction maps. Hence, the algorithm considers all possibilities of contraction maps. 

However, the algorithmic solution above is overcomplete because it could generate identical contraction maps. We discuss the complexity of the algorithm in section \ref{sec:discussions}.

\section{Holographic entropy inequalities from contraction maps}\label{sec:read-HEI}
\subsection{A greedy algorithm to generate HEIs from a contraction map}
\label{subsec:read-HEI}
In this section, we describe how to construct HEIs from a given contraction map. A contraction map is a map between $H_M$ and $H_N$, and \textit{a priori}, such a map doesn't know about the inequality. As mentioned earlier, the knowledge of the inequality is imparted to the contraction map by virtue of the boundary conditions(\ref{eq:occurence-bitstrings}), where one assigns $(n+1)$ bitstrings to $n$ monochromatic regions (and the all $\{0\}^M$ bitstring is assigned to the purifier $O$).

In principle, for an inequality involving $n$-parties (and a purifier) with $M$ terms on the LHS, the number of combinations in which one can assign the boundary conditions is $C(2^M-1,n)$\footnote{We do not get any new HEI by permutations, only the labels of regions are exchanged.}, where $C(a,b)$ denotes the binomial coefficient. Each such choice gives us a HEI (not necessarily unique). We are interested in \emph{balanced}\footnote{One may also impose the condition of superbalance to narrow down the search for facet inequalities. If we are interested only in true but not necessarily facet inequalities, superbalance is not required.} inequalities that have $M(N)$ non-trivial columns on the LHS(RHS). So we add the \emph{criterion} that after the assignment of occurrence bitstrings, all columns must correspond to some non-trivial subregion entropy. We can then store all the unique HEI candidates.\footnote{One may also be inclined to assign only those bitstrings as single-character boundary subregions whose RHS images have unique LHS pre-images. This further constrains the search space. For example, this condition, taken together with superbalance, uniquely (upto permutation of labels) determines the cyclic inequalities from their graphs.} They are all valid HEIs by the existence of contraction maps. One may further check using known extreme rays if these HEIs are potentially facets of the HEC. However, in our present work, we are not interested in that step. We summarize our algorithm to read off HEIs from the contraction map below in algorithm \ref{alg:read-HEI} and give an example of this exercise in \ref{subsec:mmi-from-map}. The computational complexity of this procedure goes as $\mathcal{O}\left((M+N)\left(2^M\right)^n\right)$ when $n<<2^M$.

\begin{algorithm}
    \caption{A greedy algorithm to generate HEIs from a contraction map.}
    \label{alg:read-HEI}
    \begin{algorithmic}[1]
        \Procedure{Contraction Map to Inequalities}{}
            \State Read the contraction map $f_{(\sigma,\iota_N)}$: $H_M \rightarrow H_N$.
            \State Assign the $\{0\}^M$ bitstring to O.
            \State For $n$-parties, generate $\texttt{num}=C(2^M-1,n)$ combinations of boundary conditions.
            \For{$i = 0;\ i < \texttt{num};\ i++$}
            \State Use boundary conditions to generate HEI $\mathcal{Q}^i_{(\sigma,\iota_N)}$.
            \If{$\mathcal{Q}^i_{(\sigma,\iota_N)}$ is balanced \textbf{and} All columns of $\mathcal{Q}^i_{(\sigma,\iota_N)}$ non-trivial}
            \State Save $\mathcal{Q}^i_{(\sigma,\iota_N)}$.
            \EndIf
            \EndFor
            \State Keep only unique $\{\mathcal{Q}^i_{(\sigma,\iota_N)}
            \}$.
        \EndProcedure
    \end{algorithmic}
\end{algorithm}

\subsubsection{Revisiting Main Problem \ref{prob:phys}}
Now we are in a position provide an answer to the main problem \ref{prob:phys} described earlier. We choose the appropriate hypercube $H_M$ corresponding to the LHS and generate all possible contraction maps $\mF= \{f_{(\sigma,\iota_N)}\}$. For each contraction map $f_{(\sigma,\iota_N)}$, we assign the same boundary conditions for occurrence vectors, faithfully representing the LHS \ref{eq:lhs-mp}, then each $f_{(\sigma,\iota_N)}$  gives a valid inequality $\mathcal{Q}^{0}_{(\sigma,\iota_N)}$ of the form \ref{eq:ineq-mp}, where the fixed superscript-$0$ refers to the fixed boundary conditions.

\subsubsection{A brief discourse on boundary conditions}

 We have discussed how to read off inequalities by choosing the number of parties $n$ and the boundary conditions for a fixed number of terms on the LHS and RHS of inequalities, given a contraction map $f$. This allows us to generate valid $n$-party HEIs. Similarly, for the contraction map $f$, we can choose other sets of boundary conditions for $n'>n$ and obtain $n'$-party HEIs. This suggests that we can get a $n$-party HEI, possibly with non-unit coefficients, by removing $n'-n$ boundary conditions from a $n'$-party inequality.

%then there should exist at least one $n'>n$-inequality that lands on a target $n$-party HEI. 

In general, consider a mapping $(H_{M'},\iota_{M'})\rightarrow(G_{N'},\iota_{N'})$, where $H_{M'}$ is a hypercube $H_{M'}$ canonically labeled by $\iota_{M'}$ corresponding to LHS and $(G_{N'},\iota_{N'})$ are the graph and labeling respectively, corresponding to the RHS of a $n'$-party HEI. Reducing from the $n'$-party inequality to a $n$-party inequality, by eliminating the $n'-n$ boundary conditions results in finding another mapping $(H_{M},\iota_M)\rightarrow(G_{N},\iota_{N})$ such that $H_M \subseteq H_{M'}$ and $G_N\subseteq G_{N'}$\footnote{For graphs $G,G'$, we denote $G\subseteq G'$ when $G$ is a subgraph of $G'$.}. This is always possible because the removal of the subset of boundary conditions corresponds to i) changing the boundary conditions without changing the graph structures, i.e., $H_M=H_{M'}$, $G_N = G_{N'}$, and $\iota_{N}=\iota_{N'}$, or ii) changing the boundary conditions with a contraction of graphs, i.e., $H_M= H_{M'}$ and $G_N \subset G_{N'}$, or $H_M\subset H_{M'}$ and $G_N \subseteq G_{N'}$\footnote{When the equalities do not hold, it corresponds to the operations where one or more columns of the table, for instance, see table \ref{tab:mmi-map}, of bitstrings are removed. Hence, the final graph is also a partial cube.}. 

In contrast, let us consider the reverse problem. That is, given a mapping $(H_{M},\iota_M)\rightarrow(G_{N},\iota_{N})$ of a $n$-party HEI, can we always find a mapping $(H_{M'},\iota_{M'})\rightarrow(G_{N'},\iota_{N'})$ such that $H_M \subseteq H_{M'}$ and $G_N\subseteq G_{N'}$? The answer is affirmative simply because there always exists a hypercube $H_{M'}$ and a partial cube $G_{N'}$ where $H_{M}$ and $G_{N}$ are isometrically embeddable. We summarize the discussion as corollary below.

\begin{corollary}
    All $n$-party HEIs with generically non-unit coefficients on the LHS having corresponding contraction maps are generated from $n'$-party HEIs with unit coefficients and non-repeating terms on the LHS\footnote{The RHS, however, is allowed to have repeating terms.} for some $n'>n$. 
\end{corollary}

We will illustrate this with an example. Consider the following five-party $(A,B,C,D,E)$ facet inequality,
\begin{equation}
 \begin{split}
& 2S(ABC) + S(ABD) + S(ABE) + S(ACD) + S(ADE) + S(BCE) + S(BDE) \geq \\ & S(AB) + S(AC) + S(AD) + S(BC) + S(BE) + S(DE) + S(ABCD) + S(ABCE) + S(ABDE)\label{eq:five-5.3}
\end{split}
\end{equation}
It has a term $2S(ABC)$, which we split, by introduce two more parties, into
$$ 2S(ABC) \rightarrow S(ABCF)+ S(ABCG) $$
Balancing the inequality on both sides, we can generate the following seven-party inequality,
\begin{equation}
\begin{split}
& S(ABCF)+ S(ABCG) + S(ABDF) + S(ABEG) + S(ACD) + S(ADE) + S(BCE) + S(BDE) \geq \\ & S(AB) + S(AC) + S(AD) + S(BC) + S(BE) + S(DE) + S(ABCDF) + S(ABCEG) + S(ABDEFG)\label{eq:seven-5.3}
\end{split}
\end{equation}
which can be easily proved using \cite{Bao:2024contraction_map}. However, this seven-party inequality, need not be a facet inequality. One can trivialize the parties $F,G$ and get back the inequality \ref{eq:five-5.3}.

\subsection{Example: Deriving the MMI from the contraction map of a star graph}
\label{subsec:mmi-from-map}
Consider the star graph shown in figure \ref{fig:graph-to-contractionmap}. We can construct this graph starting from a hypercube $H_3$ and performing the graph contraction with a partition, e.g.,
\begin{equation}
    \{000\},\{011\},\{101\},\{110\},\{001,010,100,111\}.   
\end{equation}

The resultant star graph is isometrically embeddable in a hypercube $H_4$. This embedding can be encoded in the form of a contraction map from $\{0,1\}^3$ to $\{0,1\}^4$ by (\ref{eq:graph-contractionmap}), given in table \ref{tab:mmi-map}. The labels $\{L_u\}$ and $\{R_v\}$ on the contraction map are the LHS and RHS terms of an inequality respectively, to be determined by assigning boundary conditions in $\{s_k\}$. We set our convention to assign $s_1$ to be the purifier $O$.

\begin{table}[h!]
\centering
\begin{tabular}{@{}lllllllll@{}}
\toprule
               & \textbf{$L_1$} & \textbf{$L_2$} & \textbf{$L_3$} & \multirow{9}{*}{\textbf{}} & \textbf{$R_1$} & \textbf{$R_2$} & \textbf{$R_3$} & \textbf{$R_4$} \\ \cmidrule(r){1-4} \cmidrule(l){6-9} 
\textbf{$s_1$} & \textbf{0}     & \textbf{0}     & \textbf{0}     &                            & \textbf{0}     & \textbf{0}     & \textbf{0}     & \textbf{0}     \\
$s_2$ & 0 & 0 & 1 &  & 0 & 0 & 0 & 1 \\
$s_3$ & 0 & 1 & 0 &  & 0 & 0 & 0 & 1 \\
$s_4$ & 0 & 1 & 1 &  & 0 & 0 & 1 & 1 \\
$s_5$ & 1 & 0 & 0 &  & 0 & 0 & 0 & 1 \\
$s_6$ & 1 & 0 & 1 &  & 1 & 0 & 0 & 1 \\
$s_7$ & 1 & 1 & 0 &  & 0 & 1 & 0 & 1 \\
$s_8$ & 1 & 1 & 1 &  & 0 & 0 & 0 & 1 \\ \bottomrule
\end{tabular}
\caption{A contraction map corresponding the star graph shown in figure \ref{fig:graph-to-contractionmap}.}
\label{tab:mmi-map}
\end{table}

We are interested in balanced inequalities that do not contain any trivial $L_u$ (and/or $R_v$) and thus assign the boundary conditions accordingly. For the choice of the boundary condition $(s_4, s_6, s_7)=(A,B,C)$ and its permutations, we get the MMI inequality (Eq. \ref{eq:mmi}),
\begin{equation*}
    %\label{eq:mmi}
    S(AB)+S(AC)+S(BC)\geq S(A)+S(B)+S(C)+S(ABC).
\end{equation*}
In fact, for this contraction map (table \ref{tab:mmi-map}), imposing non-triviality of $\{L_u\}$ and $\{R_v\}$ yields the MMI inequality as the only candidate, as expected from the uniqueness of the MMI contraction map. We also tabulate some other inequalities derived from the MMI map, with non-trivial column labels in table \ref{tab:mmi-ineqs}.\footnote{Note that the inequalities in table \ref{tab:mmi-ineqs} are not facet inequalities.} See Appendix \ref{app:quantum} for further discussions on relaxing the non-triviality of column labels.

\begin{table}[h!]
\centering
\begin{tabular}{@{}|c|c|@{}}
\toprule
\textbf{Parties} &
  \textbf{List of HEIs} \\ \midrule
\begin{tabular}[c]{@{}c@{}}$n=4$\\ $\{A,B,C,D\}$\end{tabular} &
  \begin{tabular}[c]{@{}c@{}}$S_{BD}+ S_{CD}+ S_{ABC} \geq S_{B}+ S_{C}+ S_{D}+ S_{ABCD},$\\ $ S_{BC}+ S_{CD}+ S_{ABD} \geq S_{B}+ S_{C}+ S_{D}+ S_{ABCD},$\\ $S_{AC}+ S_{AD}+ S_{BCD} \geq S_{A}+ S_{C}+ S_{D}+ S_{ABCD}.$\end{tabular} \\ \midrule
\begin{tabular}[c]{@{}c@{}}$n=5$\\ $\{A,B,C,D,E\}$\end{tabular} &
  \begin{tabular}[c]{@{}c@{}}$ S_{DE}+ S_{ACD}+ S_{BCE} \geq S_{C}+ S_{D}+ S_{E}+ S_{ABCDE} ,$\\ $ S_{BE}+ S_{ABD}+ S_{CDE} \geq S_{B}+ S_{D}+ S_{E}+ S_{ABCDE} ,$\\ $ S_{BD}+ S_{ABE}+ S_{CDE} \geq S_{B}+ S_{D}+ S_{E}+ S_{ABCDE} .$\end{tabular} \\ \midrule
\begin{tabular}[c]{@{}c@{}}$n=6$\\ $\{A,B,C,D,E,F\}$\end{tabular} &
  \begin{tabular}[c]{@{}c@{}}$ S_{ACE}+S_{BCF}+S_{DEF}\geq S_{C}+S_{E}+S_{F}+S_{ABCDEF}.$\end{tabular} \\ 
  \bottomrule
\end{tabular}
\caption{The holographic inequalities (up to permutations) for $n=4,5,6$ with $M=3$ and $N=4$, generated from the MMI contraction map (table \ref{tab:mmi-map}) using algorithm \ref{alg:read-HEI}.}
\label{tab:mmi-ineqs}
\end{table}

\section{Discussions}\label{sec:discussions}

\subsection{Complexity}
%'A general method for efficient embeddings of graphs into optimal hypercubes'
We discuss the complexity involved for generating all possible inequalities starting from a hypercube $H_M$.
\begin{itemize}
    \item There exists a polynomial algorithm\cite{er1988} to efficiently generate partitions of a set of cardinality $k$ with computational complexity $\mathcal{O}(k^{1.6})$. Since there are $2^M$ vertices in $H_M$, the computational complexity associated with the step \texttt{Partition generator} is $\mathcal{O}(2^{1.6M})$
    
    \item The number of all possible partitions is upper bounded by the Bell number, $B_{2^M}$. This count is redundant in the sense, multiple contractions can be geometrically \emph{equivalent} up to rotations of the RHS hypercube. So, the computational complexity for the step \texttt{Graph contraction} is $\mathcal{O}\left(2^M\left(B_{2^M}\right)\right)$. The asymptotic approximation for the logarithm of Bell numbers of order smaller than $k$ is given by \cite{Flajolet_Sedgewick_2009},
    \begin{equation}\label{eq:bellnumber}
        \ln B_k = k(\ln k - \ln \ln k -1 + o(1)). 
    \end{equation}
    where $\ln$ is the natural logarithm.

    \item The complexity of finding an embedding (if exists) for a contracted graph goes as $\mathcal{O}(|V|^2)$ where $|V|$ is the number of vertices\cite{Eppstein_2011}. Since the number of vertices is upper bounded by $2^M$, the step \texttt{Partial cube identifier} has a complexity upper bounded by $\mathcal{O}\left(\left(2^M\right)^2\right)$. Since graph contractions always reduce the number of vertices, the actual run-time is faster in most cases.
    
    \item As discussed in section \ref{subsec:read-HEI}, the computational complexity of generating possible inequalities from a given inequality is upper bounded by $\mathcal{O}\left((M+N)\left(2^M\right)^n\right)$. As proved in \ref{thm:facet-terms-lb}, we have $M+1\leq N\leq 2^{M-1}$. For facet inequalities, it is an empirical observation that $N<<2^{M-1}$. In principle, one may choose the number of parties $n$ as large as up to $2^{M-1}$, but one is usually interested in tight inequalities, which has empirically been observed to favor small $n$.
\end{itemize}
Thus, the total computational complexity to generate all possible $n$-party HEIs starting from a LHS consisting $M$ number of terms is upper bounded by $\mathcal{O}\left(B_{2^M}\left(2^M\right)^{n+3} (M+N)\right)$, where $N$ is the number of RHS terms for some inequality. 

%However, the total complexity is overestimated, although
The algorithm constructs all possible contraction maps as discussed in subsection \ref{subsec:solutions-graph}. We believe the complexity can be improved with more efficient algorithms, particularly for identifying geometrically equivalent graph contractions, and further results for bounding $n$ for tight inequalities. We leave this for future work.

\subsection{Classifications of image graphs $\Phi(H_M)$}
It is a tantalizing direction to study these image graphs and classify them into families of HEIs, shedding further light on their qualitative nature. For example, let us look at the graphs of the contraction maps for the family of cyclic inequalities (see figure \ref{fig:cyclic-graphs}). All these graphs have a similarity of symmetric structures, standing on one-leg, due to the fact that the entanglement entropy of all labeled regions appears on the RHS. In the $k=3$ case, we have three 1-dimensional edges spreading out, which gets uplifted to five 2-dimensional petals in $k=5$ case, followed by seven 3-dimensional boxes for $k=7$. Predictably, we have a symmetric arrangement of nine 4-dimensional polytopes standing on one-leg for $k=9$. This strategy appears generalizable; whenever there is a family of image graphs $G_i$ that can be inductively generated and are isometrically hypercube embeddable, it should correspond to a family of HEIs. We leave this exciting exploration of image graphs and their relation to families of polytopes for future work.

\begin{figure}[h!]
    \centering
    \begin{subfigure}{0.3\textwidth}
        \centering
        \includegraphics[width=\linewidth]{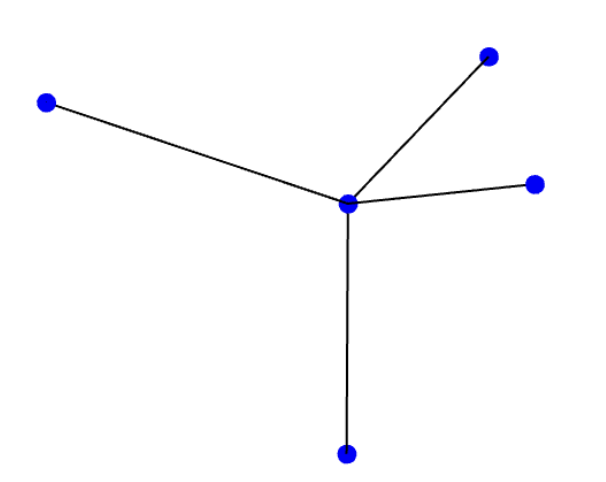}
        \caption{$k=3$}
        \label{fig:k3}
    \end{subfigure}\hfill
    \begin{subfigure}{0.3\textwidth}
        \centering
        \includegraphics[width=\linewidth]{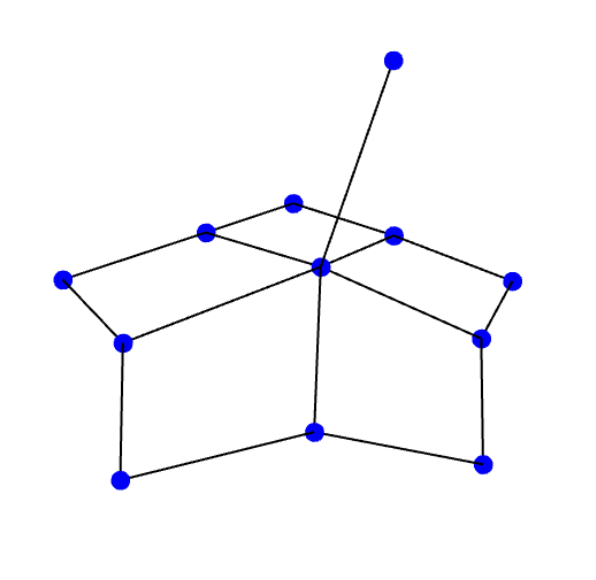}
        \caption{$k=5$}
        \label{fig:k5}
    \end{subfigure}\hfill
    \begin{subfigure}{0.3\textwidth}
        \centering
        \includegraphics[width=\linewidth]{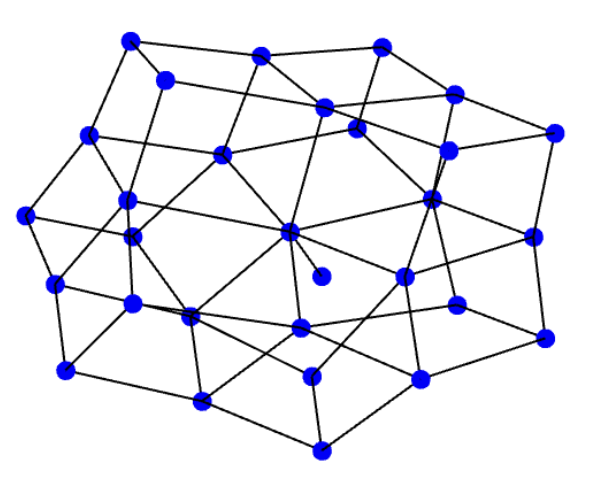}
        \caption{$k=7$}
        \label{fig:k7}
    \end{subfigure}
    \caption{The graphs associated with the first three members of cyclic inequalities.}
    \label{fig:cyclic-graphs}
\end{figure}

\subsection{Reformulating holographic entropy cones?}
The inequalities we generate in this work are more appropriately characterized by the number of unimodular entropies that appear on their LHS's, rather than their party number. While the full HEC based on fixing party is in principle recoverable in this way, it is somewhat cumbersome to do so. As such, it may be worth reformulating the HEC into one based on the number of unimodular entropies that appear, as opposed to party number, which this algorithmic approach more directly generates.

\subsection{Quantum entropy inequalities from contraction maps of HEIs}
Recall from section \ref{subsec:read-HEI}, that we imposed the non-triviality of columns in our algorithm, so as to avoid running into the ambiguity of interpreting a trivial column. But what happens when we invert this condition? It leads us to potential candidates for quantum inequalities that are valid, not just for holographic states, but all quantum states. First, consider the following proposition about holographic inequalities (see proposition \ref{thm:facet-terms-lb} for proof).
\begin{proposition}[Lower bound on $N$]\label{thm:N-lower-bound}\ 
For a facet HEI (except SA) with $M$ terms on the LHS, the number of RHS terms $N$ is bounded below by $M+1\leq N$.
\end{proposition}
One direct implication is that all such $n$-party facet inequalities (except SA) are violated by the $(n+1)$-party GHZ states (see corollary \ref{cor:GHZ} for proof). Since, a quantum inequality must always be satisfied by the GHZ state, these holographic inequalities cannot be valid candidates for quantum inequalities. Thus, we derive a necessary constraint for quantum inequalities.
\begin{corollary}\label{cor:qineq-term}
    For a quantum inequality with $M$ terms on the LHS, the number of RHS terms $N$ is bounded above by $M\geq N$.
\end{corollary}
At the expense of not interpreting the trivial columns, one can impose this new constraint from corollary \ref{cor:qineq-term} as a means to generate HEIs with $M\geq N$. These HEIs serve as potential candidates for quantum inequalities. For example, one can recover the subadditivity(\ref{eq:sa}), strong subadditivity(\ref{eq:ssa}), Araki-Lieb(\ref{eq:AL}) and weak monotonicity(\ref{eq:wm}) inequalities from the MMI contraction map, all of which are quantum inequalities. This recovery may be explained by the fact that the graph corresponding to these inequalities can be obtained by graph contractions in the MMI graph. We leave a detailed discussion about generating valid quantum inequalities from contraction maps for future work.

\subsection{What could we learn from algorithmic ``flatness''?}

Consider a set $\mF_{(M,N)}= \{f| f: \{0,1\}^M\to \{0,1\}^N\}$ of contraction maps for fixed $M$, generated by our algorithm. Given a set of disjoint boundary subregions $[n+1]$ including a purifier $O$, we can construct a constant time slice of bulk manifold $\mM_L^{b.c.}$ with at most $M$ distinct RT surfaces by giving a boundary condition($b.c.$) to a contraction map $f\in \mF_{(M,N)}$. Note that $M_L^{b.c.}$ does not depend on the choice of contraction map $f\in \mF_{(M,N)}$ since $M_L^{b.c.}$ corresponds to $2^M$ bitstrings with a boundary condition.

Let us choose a boundary condition for all contraction maps $f \in \mF_{(M,N)}$. This fixes $\mM_L^{b.c.}$. For each contraction map $f\in \mF_{(M,N)}$, we can find a bulk manifold $\mM^{b.c.}_R(f)$, with at most $N$ distinct RT surfaces. Each $\mM^{b.c.}_R(f)$ corresponds to a subset of $2^N$ bitstrings with the boundary condition. Then, the total area of RT surfaces in $\mM^{b.c.}_L$ upper bounds that of RT surfaces in $\mM^{b.c.}_R(f)$ for any $f \in \mF_{(M,N)}$. It thus implies that the number $\eta(\mM_L)$ of different ways to deform the RT surfaces in $\mM_L$ into those in $\mM_R(f)$ is at most the number of all the contraction maps, $|\mF_{(M,N)}|$.

Consider two bulk manifolds $\mM^{b.c.1}_L$ and $\mM^{b.c.2}_L$ with boundary condition $1$(b.c.1) and boundary condition $2$(b.c.2), i.e., they have distinct configurations of RT surfaces. It could be said that the bulk manifold $\mM^{b.c.1}_L$ is more ``flat''\footnote{In machine learning, \textit{flatness} of a loss surface characterizes the change in loss under the perturbations of parameters\cite{Shalev-Shwartz_Ben-David_2014}.} than $\mM^{b.c.2}_L$ if there are more ways to deform the RT surfaces in $\mM^{b.c.1}_L$ than those in $\mM^{b.c.2}_L$, i.e., $\eta(\mM^{b.c.1}_L) \geq \eta(\mM^{b.c.2}_L)$.

We leave the investigations on what the ``flatness'' could imply about properties of HEIs and bulk geometries, or vice versa.

\subsection{Further speeding up convex optimization}
It has previously been shown in \cite{Bao:2015bfa} that all HEIs are also obeyed by cut functions of general graphs. Therefore, the new classes of inequalities we explore in this work will also be respected by cut functions on graphs. We note that submodular convex optimization is known to be significantly faster than standard convex optimization methods, via the imposition of only the SSA inequality. It is worth investigating whether incorporating these more restrictive inequalities would provide further speed-up to convex optimization approaches to graph theoretic problems.

\section*{Acknowledgement}
We would like to thank Scott Aaronson, Bartek Czech, Sergio Hernandez-Cuenca, Cynthia Keeler, and Michael Walter for helpful comments. We thank Bart\l{}omiej Czech for comments on the draft. N.B. is supported by DOE ASCR, in particular under the grant Novel Quantum Algorithms from Fast Classical Transforms. K.F. is supported by N.B.'s startup funding at Northeastern University. J.N. is partially supported by the NSF under Cooperative Agreement PHY2019786 and N.B.'s startup funding at Northeastern University.

\appendix

\section{From HEIs to quantum entropy inequalities}\label{app:quantum}
In this section, we will discuss quantum entropy inequalities. First we will prove some of our claims from the text.

\begin{definition}[Tripartite form \cite{Hernandez-Cuenca:2023iqh}]\label{def:tripartite}
    An information quantity $\mathcal{Q}$ is said to be in the tripartite form if it is expressed as
    \begin{equation}\label{eq:tripartite}
        \mathcal{Q}=\sum_i -I_3(X_i:Y_i:Z_i | W_i)
    \end{equation}
    where the arguments $X_i, Y_i, Z_i, Wi \subset [N]$ are disjoint subsystems, the sum runs over any finite number of terms, and we allow for the conditioning to trivialize, $W_i = \emptyset$, in which case $I_3(X_i:Y_i:Z_i | \emptyset)= I_3(X_i:Y_i:Z_i)$ and, they are defined to be
    \begin{equation}
        I_3(X_i:Y_i:Z_i | W_i)=I_3(X_i:Y_i:Z_i W_i)-I_3(X_i:Y_i:W_i)
    \end{equation}
    and,
    \begin{equation}
        I_3(X_i:Y_i:Z_i)= X_i + Y_i + Z_i - X_i Y_i - X_i Z_i - Y_i Z_i + X_i Y_i Z_i
    \end{equation}
    We denote $I^pC^q$ for a $\mathcal{Q}$ that has $p$ number of $-I_3(X_i: Y_i: Z_i)$ and $q$ number of $-I_3(X_i:Y_i:Z_i | W_i)$ terms in the sum (\ref{eq:tripartite}).
\end{definition}
We borrow the following conjecture \ref{conj:sergio} from \cite{Hernandez-Cuenca:2023iqh}.
\begin{conjecture}\label{conj:sergio}
     All facet inequalities (except SA) are expressible in the  $I^pC^q$ form with $p\geq 1$ and $q\geq 0$.
\end{conjecture}

\begin{proposition}\label{thm:facet-terms-lb}
    For facet inequalities (except SA), $M+1 \leq N$
\end{proposition}
\begin{proof}
    According to conjecture \ref{conj:sergio}, we can write an inequality $\mathcal{Q}$ of the form $I^pC^q$ as
    \begin{equation}
        \mathcal{Q}=\sum_{i=1}^p -I_3(X_i:Y_i:Z_i) + \sum_{j=1}^q -I_3(X_i:Y_i:Z_i | W_i) := \sum_{i=1}^p I^{(i)} + \sum_{j=1}^q C^{(j)}\geq 0,
    \end{equation}
with non-trivial $W_i$.
It is simple to show that every $C^{(j)}$ has an equal number of positive and negative terms, whereas every $I^{(i)}$ has one more negative term than positive terms, where $i$ and $j$ simply labels the associated $I$ and $C$ terms respectively. Since, positive (negative) terms contribute to LHS (RHS) and $p\geq 1$, we have $M+1 \leq N$.
\end{proof}

\begin{corollary}\label{cor:GHZ}
    All $n$-party HEIs that are facets of the HEC (except SA) are violated by the $(n+1)$-party GHZ state.
\end{corollary}
\begin{proof}
    For a GHZ state, $S_{X}=S>0$ for all $X\in P(n)\backslash\emptyset$.  Since, $N\geq M+1$ for all facet inequalities (except SA), they are trivially violated.
\end{proof}

Recall the contraction map of the MMI inequality (given in table \ref{tab:mmi-mapv2} below). In this section, we use this map as an example and allow the columns to carry trivial labels such that they satisfy the necessary conditions\footnote{We leave the understanding of sufficient conditions for valid quantum entropy inequalities for future work.} for quantum entropy inequalities, namely, the existence of a contraction map and constraints on the relative number of terms appearing on two sides,
\begin{table}[h!]
\centering
\begin{tabular}{@{}lllllllll@{}}
\toprule
               & \textbf{$L_1$} & \textbf{$L_2$} & \textbf{$L_3$} & \multirow{9}{*}{\textbf{}} & \textbf{$R_1$} & \textbf{$R_2$} & \textbf{$R_3$} & \textbf{$R_4$} \\ \cmidrule(r){1-4} \cmidrule(l){6-9} 
\textbf{$s_1$} & \textbf{0}     & \textbf{0}     & \textbf{0}     &                            & \textbf{0}     & \textbf{0}     & \textbf{0}     & \textbf{0}     \\
$s_2$ & 0 & 0 & 1 &  & 0 & 0 & 0 & 1 \\
$s_3$ & 0 & 1 & 0 &  & 0 & 0 & 0 & 1 \\
$s_4$ & 0 & 1 & 1 &  & 0 & 0 & 1 & 1 \\
$s_5$ & 1 & 0 & 0 &  & 0 & 0 & 0 & 1 \\
$s_6$ & 1 & 0 & 1 &  & 1 & 0 & 0 & 1 \\
$s_7$ & 1 & 1 & 0 &  & 0 & 1 & 0 & 1 \\
$s_8$ & 1 & 1 & 1 &  & 0 & 0 & 0 & 1 \\ \bottomrule
\end{tabular}
\caption{A contraction map corresponding the star graph shown in figure \ref{fig:mmi}.}
\label{tab:mmi-mapv2}
\end{table}

In particular, we will give an example showing how one can arrive at strong subadditivity starting from the contraction map (table \ref{tab:mmi-mapv2}) corresponding to the MMI. We begin with the MMI, where one may have the the boundary conditions assigned as $(s_1, s_4, s_6, s_7)=(O,A,B,C)$. For the SSA, the boundary conditions may be changed to $(s_1, s_4, s_6, s_5)=(O,A,B,C)$ (see figure \ref{fig:mmi-to-ssa} for a graphical interpretation).
The resultant inequality is
\begin{equation}
    S(AB)+S(BC)\geq S(ABC).
\end{equation}

\begin{figure}[h!]
    \centering
    \begin{subfigure}{0.49\textwidth}
        \centering
        \includegraphics[width=\linewidth]{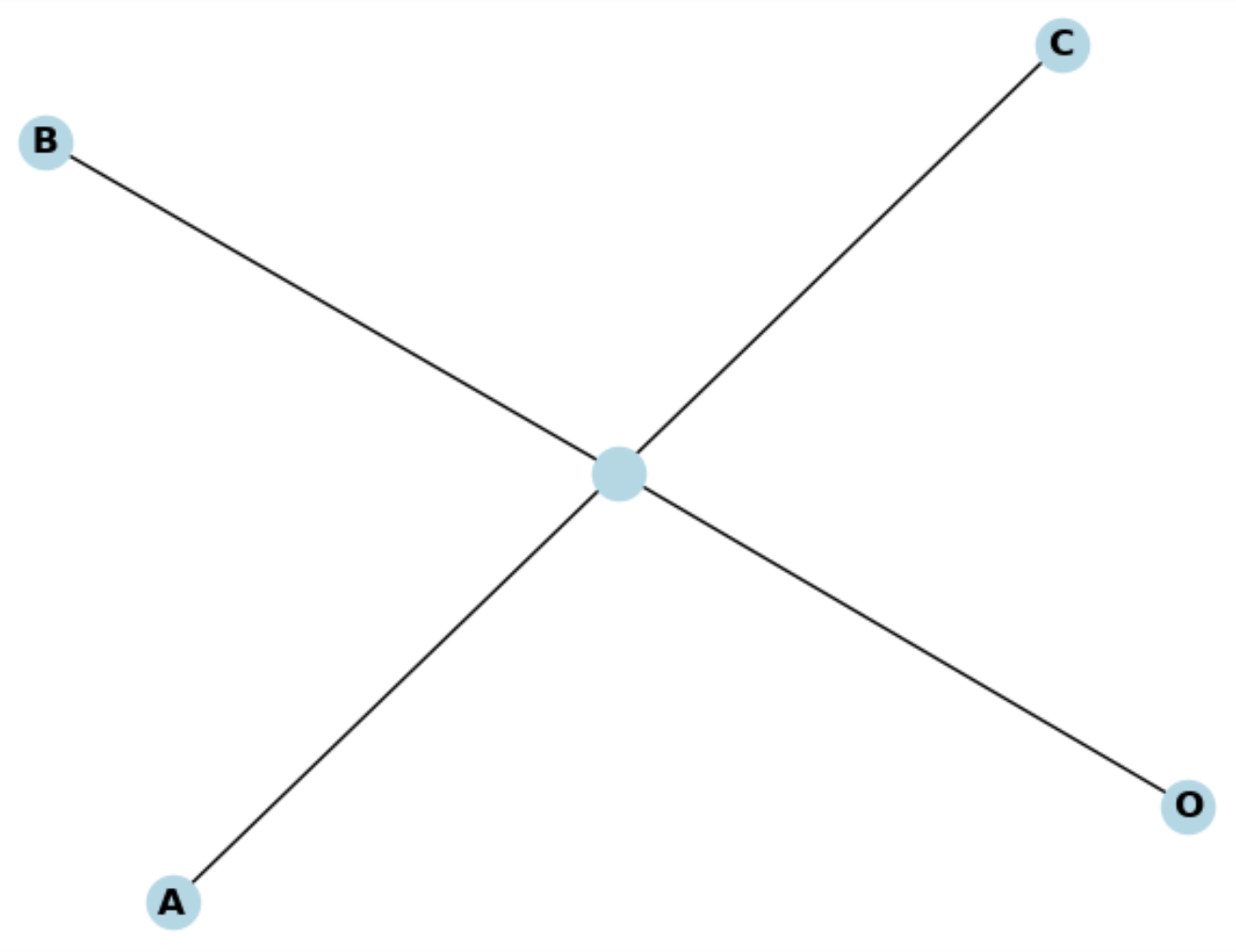}
        \caption{MMI}
        \label{fig:mmi}
    \end{subfigure}\hfill
    \begin{subfigure}{0.49\textwidth}
        \centering
        \includegraphics[width=\linewidth]{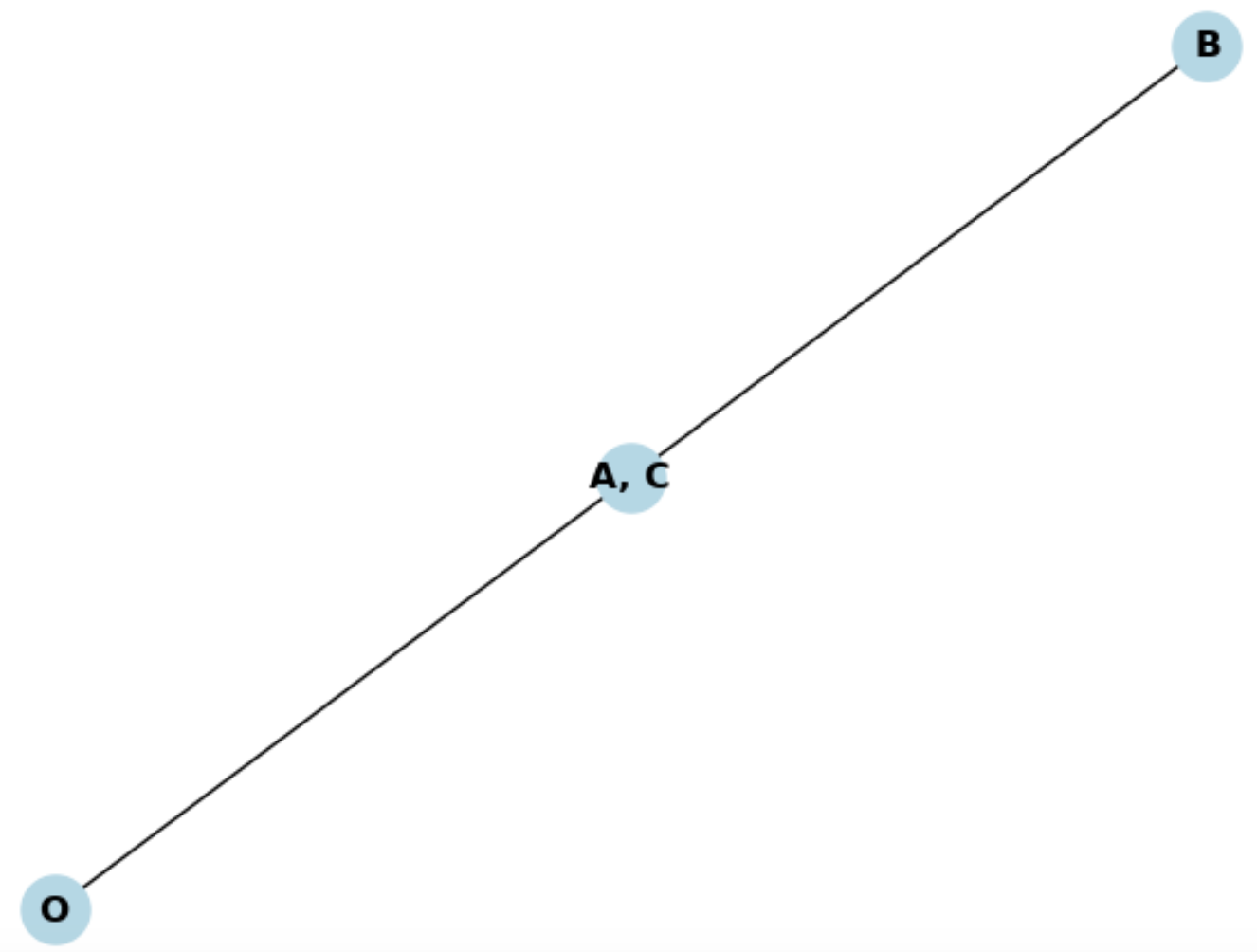}
        \caption{SSA}
        \label{fig:ssa}
    \end{subfigure}\hfill
    \caption{The star graphs corresponding to MMI (a) and SSA(b) with boundary conditions. One may start with (a), identify the vertices $A\leftrightarrow C$ and contract the edge with the central vertex of the star to get (b).}
    \label{fig:mmi-to-ssa}
\end{figure}

Similarly, we can also get the subadditivity, weak monotonicity and Araki-Lieb inequalities with the appropriate assignment of boundary conditions. We will explore this direction for higher-party contraction maps in a future work.

\section{A collection of non-facet HEIs from star graphs}
In this section, we will report a new collection of non-facet HEIs constructed from star graphs. We consider the contraction of a hypercube $H_M$ such that all bitstrings with odd number of $1$s are identified as a single vertex. This gives us a star-graph with a center vertex and $2^{M-1}$ edges coming out and joining vertices at unit Hamming distance from the center vertex. This star graph can be isometrically embedded in a $2^{M-1}$-dimensional hypercube. We will now describe the construction of the corresponding contraction map. The $\{0\}^M$ bitstring in LHS is mapped to the bitstring $\{0\}^{2^{(M-1)}}$ in RHS. Since, it is at unit distance from all the bitstrings with odd number of $1$s, those bitstrings can only map to a RHS bitstring having only one $1$ and rest $0$s. We choose the zeroth\footnote{We choose the convention where indices run from $0$ to $(2^{M-1}-1)$ from left to right. However, for the decimal conversion, we adopt the convention of increasing place values from right to left.} bit to be $1$, i.e., the RHS bitstring $\{1,0, \cdots , 0\}$. All LHS bitstrings with even number of $1$s are at unit Hamming distance from the central vertex, and twice the unit Hamming distance from each other. Therefore, the corresponding RHS bitstrings could only have at most one more $1$ at a different position from the zeroth bit. We suggest that for a LHS bitstring having a decimal equivalent $D$, the $[D/2]^{\text{th}}$ bit is assigned $1$, where $[\cdot]$ is the integer-function. This gives us a valid contraction map. We summarize the construction in algorithm \ref{alg:star-maps}.

\begin{algorithm}
    \caption{Algorithm to generate the contraction map of star graphs.}
    \label{alg:star-maps}
    \begin{algorithmic}[1]
        \Procedure{Constructing star graphs}{}
            \State A hypercube $H_M$ with vertices labeled by bitstrings $\{0,1\}^M$.
            \State Identify all vertices with bitstrings having odd number of $1$s as a single vertex.
            \State The resultant graph is a star graph with $2^{M-1}$ edges from the center vertex.
        \EndProcedure
        \State
        \Procedure{Constructing contraction maps}{}
            \State Initialize the \texttt{LHS} array to $2^M$ bitstrings canonically labeled by $\{0,1\}^M$.
            \State Initialize the \texttt{RHS} array to $2^M$ bitstrings by all $0$s, i.e., $\{0\}^{2^{(M-1)}}.$
            \For{$i = 1;\ i < 2^M;\ i++$}
            \State Assign the zeroth bit of the $i$-th \texttt{RHS} bitstring to $1$.
            \If{$i$-th \texttt{LHS} bitstring has even number of $1$s}:
            \State Compute the decimal equivalent of the $i$-th LHS bitstring, call it $D$.
            \State Assign the $[D/2]$-th bit of the $i$-th \texttt{RHS} bitstring to $1$.
            \EndIf
            \EndFor
        \EndProcedure
    \end{algorithmic}
\end{algorithm}

We now can employ algorithm \ref{alg:read-HEI} to start generating HEIs for an arbitrary number of parties $n$. In this case, however, we will relax the condition of non-triviality of columns and instead demand that all the configurations of initial conditions are drawn only from the set of vertices with even number of $1$s in the LHS bitstrings. One of the simplest example is the MMI contraction map (see table \ref{tab:mmi-mapv2}). By construction, the RHS of the inequality for $n<N$ parties has the structure,
\begin{equation}
\sum_{i=1}^{n}S_{A_i}+ S_{A_1\dots A_n}.
\end{equation}
We will now give an example with $M=5$ (see figure \ref{fig:H5-star-graph}). For example, an inequality for $n=10$ parties $(A,B,C,D,E,F,G,H,I,J)$, is as follows\footnote{This is the only balanced inequality possible with $n=10$ parties. For $n\geq 11$, no balanced inequality can be constructed using this graph. We have imposed a restriction to not choose a boundary subregion from the center vertex.},
\begin{equation}
    \begin{aligned}
        &S_{ABDG}+S_{ACEH}+S_{BCFI}+S_{DEFJ}+S_{GHIJ} \geq \\ 
        &S_{A}+S_{B}+S_{C}+S_{D}+S_{E}+S_{F}+S_{G}+S_{H}+S_{I}+S_{J}+S_{ABCDEFGHIJ}
    \end{aligned}
\end{equation}

Similarly, we can find other star-graph inequalities. In this example, we took the star-graph as an illustrative case. We will provide a complete demonstration of our framework, using all possible graphs for a fixed $M$, in a future work.

\begin{figure}[h!]
    \centering
    \includegraphics[width=\linewidth]{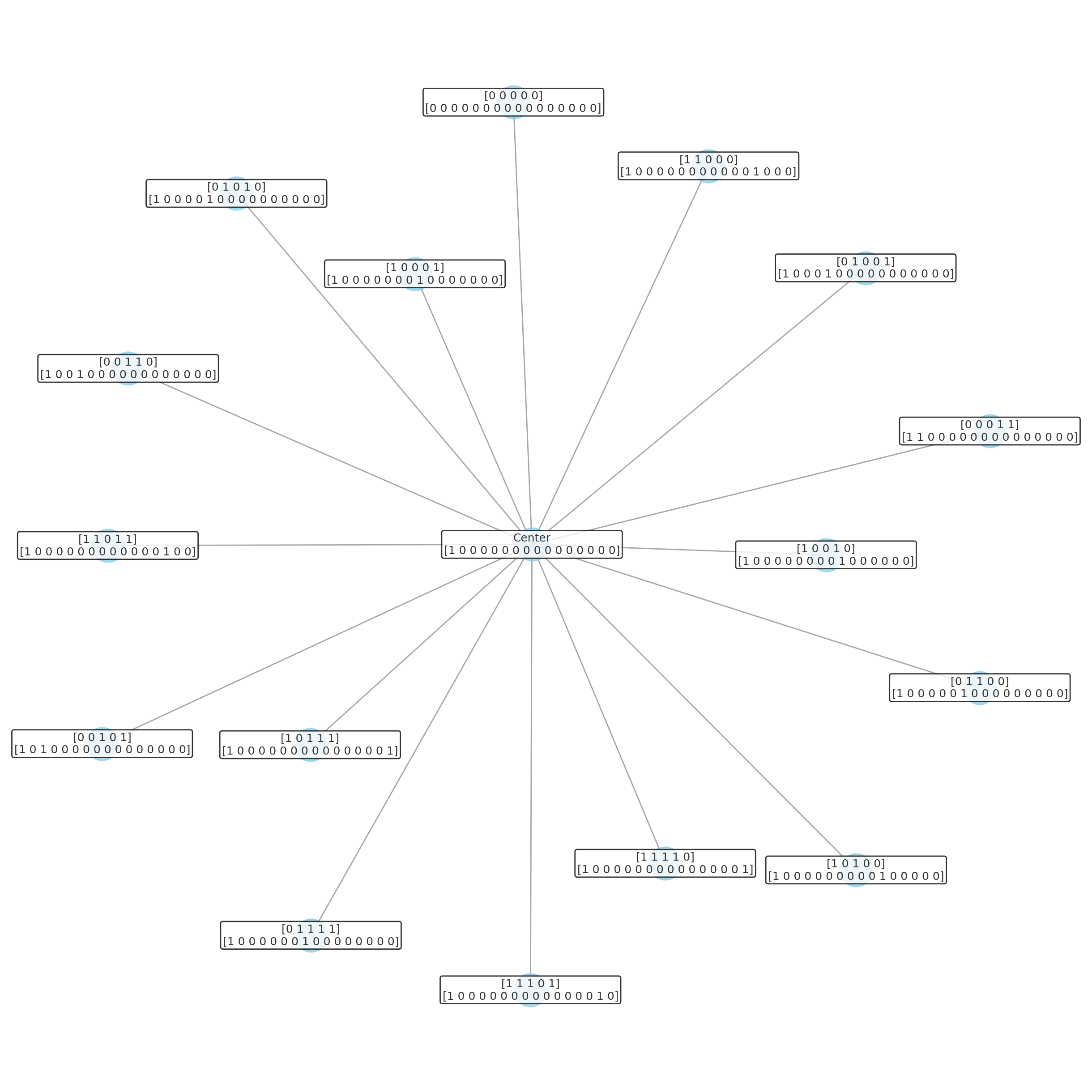}
    \caption{The star graph constructed from $H_5$. All bitstrings with odd number of $1$s are identified with the center vertex. The corresponding contraction map is also labeled.}
    \label{fig:H5-star-graph}
\end{figure}

\bibliographystyle{JHEP}
\bibliography{main}

\end{document}